\documentclass[11pt,a4paper]{article}
\usepackage{fullpage}

\usepackage{color}
\usepackage{graphicx}
\usepackage{amsmath}
\usepackage{amssymb}
\usepackage{amsthm}
\usepackage{hyperref}

\graphicspath{{./figures/}}

\newcounter{theorem}
\newtheorem{lemma}{Lemma}
\newtheorem{theorem}[lemma]{Theorem}

\newtheorem{proposition}[lemma]{Proposition}
\newtheorem{corollary}[lemma]{Corollary}

\let\geq\geqslant
\let\leq\leqslant
\def\eps{\varepsilon}
\newcommand{\IR}{\mathbb{R}}
\newcommand{\IZ}{\mathbb{Z}}

\newcommand{\poincare}{\mathbb{B}}
\newcommand{\hyperbolic}{\mathbb{H}}
\newcommand{\diam}{\mathrm{diam}}

\DeclareMathOperator{\arcosh}{arcosh}
\DeclareMathOperator{\artanh}{artanh}

\title{Coresets for Farthest Point Problems in Hyperbolic Space 
\thanks{
E. Park was supported by the National Research Foundation of Korea (NRF) grant
funded by the Korea government (MSIT) (No. RS-2024-00414849). E. Park and A. Vigneron were supported by the National Research Foundation of Korea (NRF) grant funded by the Korea government (MSIT) (No. 2022R1F1A107586911).
}
}

\author{Eunku Park\\
Department of Liberal arts and Sciences\\
DGIST, Republic of Korea\\
\texttt{parkeun9@dgist.ac.kr}
\and
Antoine Vigneron\thanks{Corresponding author}\\ 
School of Electrical and Computer Engineering\\ 
UNIST, Republic of Korea\\
\texttt{vigneron.antoine@gmail.com}
}

\begin{document}
\maketitle

\begin{abstract}
We show how to construct in linear time coresets of constant size for farthest point problems in
fixed-dimensional hyperbolic space. Our coresets provide both an arbitrarily small relative error and  additive error $\eps$. 
More precisely, we are given a set $P$ of $n$ points in the hyperbolic space $\hyperbolic^D$, where $D=O(1)$, and an error tolerance $\eps\in (0,1)$. Then we can construct in $O(n/\eps^D)$ time a subset $P_\eps \subset P$ of size $O(1/\eps^D)$ such  that for any query point $q \in \hyperbolic^D$, there is a point $p_\eps \in P_\eps$  that satisfies $d_H(q,p_\eps) \geq (1-\eps)d_H(q,f_P(q))$ and $d_H(q,p_\eps) \geq d_H(q,f_P(q))-\eps$, where $d_H$ denotes the hyperbolic metric and $f_P(q)$ is the point in $P$ that is farthest from $q$ according to this metric. 

This coreset allows us to answer approximate farthest-point queries in time $O(1/\eps^D)$ after $O(n/\eps^D)$ preprocessing time. It yields efficient approximation algorithms for the diameter, the center, and the maximum spanning tree problems in hyperbolic space.
\end{abstract}


\section{Introduction}

Nearest-neighbor searching is a fundamental problem  in computational geometry where, given
a dataset $P$ of $n$ points, we want to quickly return a point of $P$ that is closest to a query point $q$. If interpoint distances can be computed in $O(1)$ time, a query can be answered by brute force in $O(n)$ time, so the goal is to design data structures that answer queries in sublinear time after some preprocessing. For instance, in the Euclidean plane $\IR^2$, nearest-neighbor queries can be answered in $O(\log n)$ time after $O(n \log n)$ preprocessing time by computing the Voronoi diagram of $P$~\cite{berg2008computational}.
Unfortunately, the worst case complexity of the the Voronoi diagram in the $D$-dimensional
Euclidean space $\IR^D$ is $\Theta(n^{\lceil D/2 \rceil})$, so this approach is often impractical in dimension $D \geq 3$. 

In order to address this issue, an approximate version of nearest-neighbor searching, called 
{\it approximate nearest-neighbor searching} (ANN) has been considered where, instead of returning the closest point to $q$, we return a point whose distance from $q$ is at most $(1+\eps)$ times the optimal, for a given relative error tolerance $\eps \in (0,1)$. Arya et al. showed that in fixed-dimensional Euclidean space (i.e. in $\IR^D$ where $D=O(1)$), ANN queries can be answered in $O(\log(n)/\eps^D)$ time after $O(n \log n)$ preprocessing time~\cite{Arya98}.

A related problem is farthest-neighbor searching: Preprocess $P$ so that a point in $P$ that is farthest from a query point $q$ can be found efficiently. This problem can be solved exactly using the {\it farthest-point Voronoi diagram}, but unfortunately, in the worst case, it has the same size $\Theta(n^{\lceil D/2 \rceil})$ as the Voronoi diagram in dimension $D$.

Agarwal et al.~\cite{agarwal1992farthest} gave an efficient data structure for $(1-\eps)$-approximate farthest neighbors in fixed dimension: After $O(n/\eps^{(D-1)/2})$ preprocessing time, a point whose distance from $q$ is at least $(1-\eps)$ times the maximum is returned in $O(1/\eps^{(D-1)/2})$ time. The idea is to compute the set $P_\eps$ of extreme points of $P$ along 
$O(1/\eps^{(D-1)/2})$ directions distributed (roughly) uniformly on the unit sphere. This set $P_\eps$ has size $O(1/\eps^{(D-1)/2})$, and for any query point $q$, an approximate farthest point is in $P_\eps$. Thus, it is a {\it coreset} for farthest point problems: It is a small subset of $P$, such that if we want to answer an approximate farthest-neighbor query on $P$, we can return a point in $P_\eps$ that is farthest from the query point. Coresets have been used for approximating the solutions of several computational geometry problems~\cite{HPbook}.

In this paper, we give a coreset of size $O(1/\eps^D)$ for farthest-point problems in 
fixed-dimensional hyperbolic space $\hyperbolic^D$.  This space is the unique $D$-dimensional
Riemannian manifold with sectional curvature -1 at every point. It has several
isometric models such as the hyperboloid model and the Poincar\'e half-space 
model~\cite{ratcliffe1994foundations}.
For convenience, we will use the {\it Poincar\'e ball} model $(\poincare^D,d_H)$ where $\poincare^D$ is
the open unit ball of $\IR^D$ centered at the origin $O$, and $d_H$ is the metric defined
as follows. (See Section~\ref{sec:ball} for a more detailed introduction to the Poincar\'e ball model.)

The hyperbolic length $L_H(\Gamma)$ of a curve $\Gamma$ in $\poincare^D$ is given by the integral 
$L_H(\Gamma)=\int_{\Gamma} ds$ where 
$$ds^2=4\frac{dx_1^2+\dots+dx_D^2}{\left(1-(x_1^2+\dots+x_D^2)\right)^2}.$$
(In the Euclidean case, $ds$ is given by the relation $ds^2=dx_1^2+\dots+dx_D^2$.) For any two distinct points $u,v \in \poincare^D$, the curve $[u,v]$ of minimum hyperbolic length between $u$ and $v$ is called a {\it geodesic}. The hyperbolic distance $d_H(u,v)$ is the hyperbolic length of $[u,v]$.

\begin{figure}
	\centering
	\includegraphics{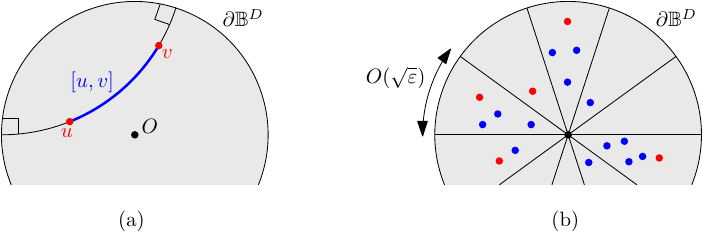}
	\caption{(a) A geodesic $[u,v]$ in the Poincar\'e ball model $\poincare^D$.\label{fig:poincareball}
		(b) Cone based approach: The input points are in blue. We add the red points to the coreset.}
\end{figure}

The hyperbolic distance $d_H(u,v)$ can be obtained by a closed formula. 
(Equations~\eqref{eq:defdH1}--\eqref{eq:defdH2} in Section~\ref{sec:ball}.) It can be shown that the geodesics are arcs of circles that are orthogonal to the unit sphere $\partial \poincare^D$.
(See \figurename~\ref{fig:poincareball}a.) This metric space $(\poincare^D,d_H)$ has several properties that make it very different from the Euclidean space. For instance, the volume of a ball is exponential in its radius, instead of being polynomial, and thus some packing or discretization arguments used in the Euclidean case break down. On the other hand, triangles in hyperbolic space are {\it thin}: There is a constant $\delta$ such that for any $u,v,w \in \poincare^D$, and for any point $t \in [u,w]$, there is a point in $[u,v] \cup [v,w]$ that is at distance at most $\delta$ from $u$. (For a more precise statement, see Section~\ref{sec:geogromov}.) 

One motivation for studying algorithms in hyperbolic spaces is that computer networks, and in particular the internet, are known to have better embeddings into hyperbolic space than into Euclidean spaces~\cite{1354510}. Recently, connections have also been found between machine learning and hyperbolic geometry, in particular in the context of hyperbolic neural 
networks~\cite{NEURIPS2018,Nature}.

\paragraph{Our results.}

We assume that we are in fixed dimension $D=O(1)$, and that $\eps \in (0,1)$.
We are given an input set $P$ of $n$ points in $\poincare^D$. 
For any point $q \in \poincare^D$, we denote by $f_P(q)$ a point in $P$ that is
farthest from $q$:
\[
	d_H(q,f_P(q))=\max_{p \in P} d_H(q,p).
\]
We show (Theorem~\ref{th:faralgo}) how to construct in $O(n/\eps^D)$ time a coreset
$P_\eps \subset P$ such that for any query point $q$, there is a point $p_\eps \in P_\eps$
that satisfies
\[
	d_H(q,p_\eps) \geq (1-\eps)\max_{p \in P} d_H(q,f_P(q)) \quad \text{ and } \quad
	d_H(q,p_\eps) \geq \max_{p \in P} d_H(q,f_P(q)) - \eps. 
\]
Thus, our coreset provides both a $(1-\eps)$-{\it multiplicative} and  an $\eps$-{\it additive} error bound. It is obviously impossible to obtain a constant additive error in Euclidean space.

It directly allows us to answer approximate farthest-neighbor queries in time $O(1/\eps^D)$, after $O(n/\eps^D)$ preprocessing time, by returning the point in $P_\eps$ that is farthest from the query point $q$. Again, we obtain both an additive and a multiplicative error bound.

Our coreset yields approximation algorithms for other problems in $\hyperbolic^D$.
\begin{itemize}
\item
	The {\it diameter} $\diam(P)$ of $P$ is the maximum interpoint distance in $P$, so  $\diam(P)=\max_{u,v \in P} d_H(u,v)$. We obtain a $(1-\eps)$-multiplicative and $\eps$-additive error to $\diam(P)$ in $O(n/\eps^D)$ time by performing an approximate farthest neighbor for each point in $P$, and returning the largest result.
\item A {\it center} of $P$ is a point $c \in P$ such that $d_H(c,f_P(c))$ is minimized. We can find an approximate radius (both additive and multiplicative) in $O(n/\eps^D)$ by performing an approximate farthest neighbor query on every point in $P$,
\item We can compute an approximate {\it bichromatic closest pair} in $O(n/\eps^D)$ time and an approximate {\it maximum spanning tree} in $O(n \log^2(n)/\eps^D)$ time using the reduction by Agarwal et al.~\cite{agarwal1992farthest}. For the maximum spanning tree problem, we only get a $(1-\eps)$-multiplicative approximation, and not an additive error.
\end{itemize}

\paragraph{Our approach.} There are two main cases in our construction. In Section~\ref{sec:smalldiameter}, we deal with small diameter input. We observe that, after applying a suitable isometry, the hyperbolic distance is within a constant factor from the Euclidean metric. Then we obtain $P_\eps$ by a bucketing approach, using a regular grid of mesh size $O(\eps)$. This is a well-known approach in geometric approximation algorithms~\cite{HPbook}.

In Section~\ref{sec:far}, we consider large diameter input, where $\diam(P) \geq 5$.
In order to handle query points at distance $O(1)$ from $O$, we use a grid-based approach similar
to the small diameter case, which yields $O(1/\eps^D)$ coreset points.
At large scale, a simple grid-based discretization fails as the volume of a hyperbolic ball grows exponentially with its radius. Instead, we use a cone-based approach. 

We partition the space around $O$ into $O(1/\eps^{(D-1)/2})$ cones of angular diameter $O(\sqrt{\eps})$. In each cone, we add to our coreset the point that is farthest from $O$. (See \figurename~\ref{fig:poincareball}b.) Then we show, based on geometric arguments presented in Section~\ref{sec:geopoinc}, that the points we added to our coreset can handle query points $q$ such that the geodesic $[q,f_P(q)]$  passes close enough to $O$.

In order to reduce to the case where $[q,f_P(q)]$ is close to $O$, we apply a set of $O(1/\eps^{(D-1)/2})$ isometries of $\poincare^D$, each isometry mapping to the origin $O$ a point taken from an approximately uniform sample of a constant radius sphere centered at $O$. For each such translation, we apply the cone-based approach above. It yields $O(1/\eps^{D-1})$ more coreset points. Our correctness proof uses properties of Gromov-hyperbolic space, to which we give an introduction in Section~\ref{sec:Gromov}, including a linear-time algorithm by Chepoi et al.~\cite{chepoi2008diameters} for approximating the diameter of a Gromov-hyperbolic space. 

\paragraph{Comparison with previous work.} Euclidean nearest-neighbor searching has been extensively studied. For instance, Clarkson gave a data structure with $O(\log n)$ query time and $\tilde O(n^{\lceil D/2 \rceil})$ preprocessing time~\cite{Clarkson88}. As mentioned above, 
Arya et al.~\cite{Arya98} gave a data structure for the approximate version (ANN) with approximation factor $(1+\eps)$, query time $O(\log(n)/\eps^D)$ time and $O(n \log n)$ preprocessing time~\cite{Arya98}.
A coreset other than the input set itself is not possible for ANN in Euclidean space, as any input point $p$ is the result of the query $q=p$.

As mentioned above, Agarwal et al.~\cite{agarwal1992farthest} gave a coreset of size $O(1/\eps^{(D-1)/2})$ for farthest-neighbor queries in the Euclidean space $\IR^D$. Their coreset only provides a $(1-\eps)$-multiplicative error, as the only coreset with constant additive error in $\IR^D$ is $P_\eps=P$ in the worst case.  
Another coreset for approximate farthest-neighbor searching was given recently by de Berg and Theocharou~\cite{BergT24}. They showed how to construct a coreset of size $O(1/\eps^2)$ for farthest-neighbor queries in a simple polygon, using the geodesic distance. 
Pagh et al.~\cite{pagh2015approximate} studied approximate farthest-neighbor searching in high dimension (i.e. when we do not assume that $D=O(1)$), with query time $\tilde O(n^{1/c^2})$
and space usage $O(n^{1/(2c^2)})$ in dimension $O(\log n)$.

Recently, two-dimensional Voronoi diagrams~\cite{gezalyan} and farthest-point Voronoi diagrams~\cite{SongJA25} have been studied under the Hilbert metric. This metric generalizes the Cayley-Klein model of $\hyperbolic^D$ to any convex polygon, instead of the unit disk.

To the best of our knowledge, there is no previous work on farthest-neighbor searching in the hyperbolic space $\hyperbolic^D$ when $D \geq 3$. However, Chepoi et al.~\cite{chepoi2008diameters} studied related problems in the more general setting of Gromov $\delta$-hyperbolic spaces. A metric space $(X,d)$ is
$\delta$-hyperbolic if for any $t,u,v,w \in M$, the two largest sums among $d(t,u)+d(v,w)$, $d(t,v)+d(u,w)$ and $d(t,w)+d(u,v)$ differ by at most $2\delta$. The hyperbolic space $\hyperbolic^D$ is Gromov hyperbolic, with $\delta=\log 3$. (See Section~\ref{sec:Gromov} for a brief introduction to Gromov hyperbolic spaces, with an equivalent definition.) 
Chepoi et al.~\cite{chepoi2008diameters} showed that the diameter of a $\delta$-hyperbolic space can be approximated with additive error $2\delta$ in linear time. With our approach, the diameter of a subset of $\hyperbolic^D$ can be approximated in $O(n/\eps^D)$ time with an additive error $\eps$, so we obtain a better approximation error, but in a more restricted setting. They provide a similar result for approximating the radius and the eccentricities of all points, which we can improve in the same way, still in the special case of $\hyperbolic^D$.

Approximate near-neighbor searching has also been studied in the context of hyperbolic spaces. Krauthgamer and Lee~\cite{krauthgamer2006algorithms} gave a data structure for ANN in a special case of $\delta$-hyperbolic spaces, that applies to $\hyperbolic^D$, with $O(\delta)$ additive error, $O(n^2)$ space usage and $O(\log^2 n)$ query time. More recently, 
Kisfaludi-Bak and Wordragen~\cite{kisfaludi2024quadtree} gave an ANN data structure for $\hyperbolic^D$ with $(1+\eps)$ factor approximation, $O(n \log(1/\eps)/\eps^D)$ size
and $O(\log n \log(1/\eps)/\eps^D)$ query time. In our previous work~\cite{park2025embeddings},  we gave a data structure for ANN in $\hyperbolic^D$ with $O(1)$ additive error, $O(\log n)$ query time and $O(n\log n)$ construction time.

\section{Preliminaries}

In this paper, $P = \{ p_{1}, ..., p_{n} \}$ denotes a set of $n$ points in a metric space $(X,d)$.
This metric space will be either the Poincar\'e ball model $(\poincare^D,d_H)$ 
of the $D$-dimensional hyperbolic space $\hyperbolic^D$
(Section~\ref{sec:ball}), or a Gromov hyperbolic space (Section~\ref{sec:Gromov}).
For any point $q \in X$, we denote by $f_P(q)$ a point in $P$ that is farthest to $q$, hence
\[
	d(f_P(q),q)=\max_{p \in P} d(p,q).
\]
The {\it diameter} of $P$ is the maximum distance
$
\diam(P) = \max_{p,q \in P} d(p,q),
$
and a {\it diametral pair} is a pair of points $a^*,b^* \in P$ such that $\diam(P)=d(a^*,b^*)$.

The metric space $(X,d)$ is a {\it geodesic space} if, 
between any two points $p$ and $q$, there exists a shortest path $[p,q] \subset X$, called a {\it geodesic}.
More precisely, there is an isometry $\gamma:[0,d(p,q)] \to X$ such that $\gamma(0)=p$ and
$\gamma(d(p,q))=q$.
For any 3 points $p,q,r \in X$, the {\it geodesic triangle} $[p,q,r]$ is the union of three geodesics
$[p,q]$, $[q,r]$, $[r,p]$, called its {\it sides}. 

In this paper, we use the natural logarithm, denoted by $\log(\cdot)=\log_e(\cdot)$.

\subsection{The Poincar\'e Ball Model}\label{sec:ball}
The hyperbolic space $\hyperbolic^D$ is the $D$-dimensional space of constant sectional curvature -1.
Several isometric models of $\hyperbolic^D$ have been considered, including the
Poincar\'e half-space model and the hyperboloid model. We will use
the {\it Poincar\'e ball model}, also called {\it conformal ball model}~\cite{ratcliffe1994foundations}.
In this model, $\hyperbolic^D$ is identified with the open unit ball
$
\poincare^{D} = \{ x \in \mathbb{R}^{D} : \| x \| < 1 \}.
$
It is equipped with the {\it hyperbolic metric} $d_H$, given by the expressions
\begin{align} \label{eq:defdH1}
d_H (u,v)  
& = \arcosh \left( 1 + \frac{2 \| u-v\|^{2}}{(1 - \| u \|^{2}) (1 - \| v \|^{2})} \right) \\ 
\label{eq:defdH}
& = 2 \log \frac{\|u-v\|+\sqrt{\|u\|^2\|v\|^2-2u \cdot v+1}}{\sqrt{(1 - \| u \|^{2})(1 - \| v \|^{2})}},
\end{align}
where $\| \cdot \|$ is the Euclidean norm. 
In the special case where  $\|v\|=r$, we have
\begin{equation}\label{eq:defdH2}
d_H (O,v) = 2 \artanh r = \log\left(\frac{1+r}{1-r}\right).
\end{equation}
Using this metric $d_H$, 
geodesics are arcs of circles orthogonal to the boundary sphere $\partial \poincare^D$. 
The distance $d_H(u,v)$ goes to infinity when $\|v\|$ goes to 1, so the boundary $\partial \poincare^D$ of the Poincar\'e ball can be regarded as the set of points at infinity.

\paragraph{Isometries.}
We will need to be able to change the center $O=(0, \dots,0)$ of the Poincar\'e ball in order
to simplify our calculations. More precisely, 
let $h$ be a point in $\poincare^D$.  The {\it hyperbolic translation} $\tau_h$ is an isometry
of $(\poincare^D,d_H)$ such that $\tau_h(O)=h$. It is given by the 
expression~\cite{ratcliffe1994foundations}:
\begin{equation*}
\tau_{h} (u) = \frac{(1- \| h \|^{2})u + ( \| u \|^{2} + 2 \langle u, h \rangle + 1) h}
{\| h \|^{2} \| u \|^{2} + 2 \langle u, h \rangle + 1}
\end{equation*}
where $\langle u, h \rangle$ denotes the standard Euclidean inner product in $\mathbb{R}^{D}$.
This map can be computed in constant time, as well as its inverse $\mu_h=\tau_h^{-1}$. Then
$\mu_h$ is an isometry that maps $h$ to $O$, and can be computed in constant time. It follows that:
\begin{proposition}\label{prop:translation}
		Let $h$ be an arbitrary point in $\poincare^D$ where $D=O(1)$. 
		There is an isometry $\mu_h$ of $(\poincare^D,d_H)$ 
		such that $\mu_h(h)=O$, and for any point $p \in \poincare^D$, we can compute $\mu_h(p)$ in 
		$O(1)$ time.
\end{proposition}

\section{Small Diameter Input}\label{sec:smalldiameter}

In this section, we first deal with the case where the input point set $P$ has constant
diameter. So we first compute $\Delta_{1}=d_H(p_1,f_P(p_1))$, and thus 
$
	\Delta_1 \leq \diam(P) \leq 2\Delta_1.
$
We assume that $\Delta_1 \leq 5$, and we apply the isometry $\mu_{p_1}$ so that $p_1$
is the origin $O$ of the Poincar\'e ball $\poincare^D$.
We first observe the following
\begin{lemma}\label{lem:half}
For any $q \in \poincare^{D}$, we have $d_H(q,f_P(q)) \geq \Delta_{1}/2$.
\end{lemma}

\begin{proof}
Let $q_1=f_P(p_1)$. By the triangle inequality, 
$\Delta_1=d_{H} (p_1, q_1) \leq d_{H} (p_1,q) + d_{H} (q,q_1)$ holds. 
It implies that $d_H(p_1,q) \geq \Delta_1/2$ or $d_H(q,q_1) \geq \Delta_1/2$,
and the result follows.
\end{proof}

Let $B_H(O,\Delta_1 )$ be the ball centered at $O$ with hyperbolic radius $\Delta_1$, and thus
with Euclidean radius $R=\tanh(\Delta_1/2)\leq \tanh(5/2)$. 

\begin{lemma}\label{lem:smalldiameter1}
	For any two points $u,v \in B_H(O,\Delta_1)$, we have
	\[
		2\|uv\|  \leq d_H(u,v)  < 76\|uv\|.
	\]
\end{lemma}
\begin{proof}	
	Let $[u,v]$ be the geodesic from $u$ to $v$ and $\overline{uv}$ be the line
	segment from $u$ to $v$. Then we have
	\[ 		2\|uv\|  =  \int_{\overline{uv}} 2\|dx\| 
			 \leq \int_{[u,v]} 2\|dx\| 
			\leq \int_{[u,v]} \frac{2\|dx\|}{1-\|x\|^2} 
			 =d_H(u,v) 
	\]
	and
	\begin{equation*}
		d_H(u,v)  = \int_{[u,v]} \frac{2\|dx\|}{1-\|x\|^2} \leq  \int_{\overline{uv}} \frac{2\|dx\|}{1-\|x\|^2} 
			  \leq  \frac{2}{1-R^2} \int_{\overline{uv}} \|dx\| = \frac{2}{1-R^2} \|uv\| < 76 \|uv\|.
	\end{equation*}
\end{proof}

The $\alpha$-grid $G_\alpha$ in $\IR^D$ is the discrete set of points whose coordinates are
multiples of $\alpha$, so we have $G_\alpha=\alpha \IZ^D$. For each point $p \in (\IR^+)^D$,
the point of $G_\alpha$ obtained by rounding down each coordinate to the nearest multiple of
$\alpha$ is at Euclidean distance at most $\alpha \sqrt D$ from $p$ and is closer to $O$ than $p$ is.
It follows that
\begin{proposition}\label{prop:grid}	
	For each point $p \in B_H(O,\Delta_1)$, there is a point $p_\alpha \in G_\alpha \cap B_H(O,\Delta_1)$
	such that $d_H(p,p_\alpha) \leq 76 \alpha \sqrt D$.
\end{proposition}

So we construct our coreset $P_\eps$ for farthest-point problems as follows. We set
$\alpha=\eps \Delta_1/ (304 \sqrt D)$. We use a bucketing approach. 
We place any two points $p,p' \in P$ such that $p_\alpha=p'_\alpha$ 
in the same bucket, and keep only one point from each bucket, obtaining a subset $P_\eps \subset P$
such that for any point in $p \in P$, there is a point $p'' \in P_\eps$ that satisfies $p_\alpha=p''_\alpha$.
This construction can be done in $O(n)$ time, and we have $|P_\eps|=O(1/\alpha^{D})=O(1/\eps^D)$.

For any query point $q \in \poincare^D$, let $q'=f_P(q)$. Then there is a point $q'' \in P_\eps$
such that $q'_\alpha=q''_\alpha$. It follows from Proposition~\ref{prop:grid} that 
\[
	d_H(q',q'') \leq d_H(q',q'_\alpha)+d_H(q''_\alpha,q'') \leq 152\alpha\sqrt{D}=\eps\Delta_1/2.
\]
Then by Lemma~\ref{lem:half}, 
\[
	d_H(q,q'') \geq d_H(q,q')-d_H(q',q'') \geq d_H(q,q')- \eps \Delta_1/2 \geq (1-\eps) d_H(q, q')
\]
and thus $q''$ is an $\eps$-approximate farthest neighbor of $q$ in $P$. So we proved the following.

\begin{lemma}\label{lem:small}
	Let $D$ be a fixed integer and $0 < \eps < 1$.
	Let $P=\{p_1,\dots,p_n\}$ be a set of $n$ points in $\poincare^D$ 
	such that $\Delta_1=d_H(p_1,f_P(p_1))$ satisfies $\Delta_1 \leq 5$.
	We can construct in $O(n)$ time a coreset $P_\eps \subset P$ for farthest-point queries of 
	size $|P_\eps|=O(1/\eps^D)$. In particular, for any query point $q \in \poincare^D$,
	there is a point $p_{\eps} \in P_{\eps}$ such that 
	$
		d_H(q,p_{\eps}) \geq (1-\eps) d_H(q,f_P(q)).
	$
\end{lemma}

\section{Geometric Lemmas for the Poincar\'e Ball Model}  \label{sec:geopoinc}

In this section, we gather a few lemmas that will be needed for our coreset construction
when $\diam(P)>5$.
We will use the following inequalities, which can be obtained by 
Taylor expansion.

\begin{lemma}\label{lem:Taylor}
(a) $\cos x \geq 1 - \frac{x^{2}}{2}$ for all $x$.
(b) $\sqrt{1-x} \geq 1 - x$ for all $0 \leq x \leq 1$.
(c) $\frac{1}{1-x} \leq 1+2x$ for all $0 \leq x \leq 1/2$.
(d) $\log (1+x) \leq x$ for all $x$.
\end{lemma}

For a geodesic triangle  $[O,u,v]$ in $\poincare^D$ where $Ou$ and $Ov$  make an angle $\pi-\theta$, the lemma below shows that the triangle inequality is within an additive error $\theta^2$ from equality. (See \figurename~\ref{fig:circles}.)

\begin{figure}
	\centering
	\includegraphics{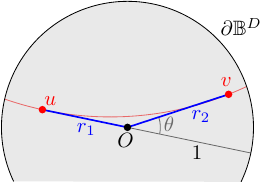}
	\caption{Lemma~\ref{lem:theta}.\label{fig:circles}}
\end{figure}

\begin{lemma}\label{lem:theta}
Let $u,v$ be two points in $u,v \in \poincare^{D}$ and let  $\theta = \pi-\angle uOv$.
If $\theta < 1$,  then we have $d_H (u, O) + d_H (O, v)  \leq  d_H (u,v) + \theta^{2}$.
\end{lemma}

\begin{proof}
Let $r_1=\|u\|$ and $r_2=\|v\|$. We define 
\[f(r_1, r_2, \theta) = d_H (u,O) + d_H (O,v) - d_H (u,v).\]
Then by Equations~\eqref{eq:defdH} and~\eqref{eq:defdH2}, we have
\begin{align*}
f(r_1, r_2, \theta) & =  \log \left(\frac{1+r_1}{1-r_1}\right) + 
	 \log \left(\frac{1+r_2}{1-r_2}\right) \\
	& \quad - 2 \log \left(\frac{
				\sqrt{(r_1+r_2\cos \theta)^2+r_2^2 \sin^2 \theta}+
				\sqrt{r_1^2r_2^2+2r_1r_2\cos \theta+1}
				}{
				\sqrt{(1-r_1^2)(1-r_2^2)}
				}
	\right) 
	\\
	& = 2\log \left ( \frac{(1+r_1) (1+r_2)}{ \sqrt{r_1^{2} + r_2^{2} +2 r_1 r_2 \cos \theta} + \sqrt{r_1^{2} r_2^{2} + 2r_1 r_2 \cos \theta + 1}} \right )
\end{align*}

We first give lower bounds for the two parts of the denominator.
\begin{align*}
\sqrt{r_1^{2} + r_2^{2} +2 r_1 r_2 \cos \theta} 
	&\geq \sqrt{r_1^{2} + r_2^{2} +2 r_1 r_2 \left ( 1 - \frac{\theta^{2}}{2} \right )} 
	 & \text{by Lemma~\ref{lem:Taylor}a}\\
		&= \sqrt{(r_1+r_2)^{2} - r_1 r_2 \theta^{2}} \\
		&= (r_1+r_2) \sqrt{1 - \frac{r_1 r_2}{(r_1 + r_2)^{2}}\theta^{2}} \\
		& \geq (r_1+r_2) \sqrt{1 - \frac 1 4 \theta^{2}} &  \text{because $r_1^2+r_2^2 \geq 2r_1r_2$} \\
		&\geq (r_1 + r_2) \left ( 1 - \frac 1 4 \theta^{2} \right )
		& \text{by Lemma~\ref{lem:Taylor}b}
\end{align*}
\begin{align*}
\sqrt{r_1^{2} r_2^{2} + 2r_1 r_2 \cos \theta + 1} &\geq \sqrt{r_1^{2} r_2^{2} + 2r_1 r_2 \left ( 1 - \frac{\theta^{2}}{2} \right ) + 1} \\
		&= \sqrt{(r_1 r_2 + 1)^{2} - r_1 r_2 \theta^{2}} \\
		&= (r_1 r_2 + 1) \sqrt{1 - \frac{r_1 r_2}{(r_1 r_2 + 1)^{2}}\theta^{2}} \\
		& \geq (r_1 r_2 + 1) \sqrt{1 - \frac{1}{4}\theta^{2}}  & \text{because $r_1^2 r_2^2 +1 \geq 2r_1r_2$} \\
		&\geq (r_1 r_2 + 1) \left ( 1 - \frac{1}{4} \theta^{2} \right )
		& \text{by Lemma~\ref{lem:Taylor}b
		}
\end{align*}
It follows that 
\begin{align*}
f(r_1,r_2,\theta) & = 2\log
\left ( \frac{(1+r_1) (1+r_2)}
{\sqrt{r_1^{2} + r_2^{2} +2 r_1 r_2 \cos \theta} + \sqrt{r_1^{2} r_2^{2} + 2r_1 r_2 \cos \theta + 1} } \right ) \\
& \leq 2 \log \left ( \frac{(1+r_1) (1+r_2)}
{(r_1 r_2 + r_1 + r_2 + 1)\left(1-\frac 1 4 \theta^{2}\right)} \right ) \\
& = 2 \log\left(\frac{1}{1-\frac 1 4 \theta^2} \right) \\
& \leq 2\log\left(1+\frac 1 2 \theta^2\right) & \text{by Lemma~\ref{lem:Taylor}c} \\
& \leq \theta^2. &  \text{by Lemma~\ref{lem:Taylor}d}
\end{align*}
\end{proof}

When a geodesic connecting two points at infinity goes close to 
the origin $O$, the lemma below gives bounds on the angle they form about $O$. 
(See \figurename~\ref{fig:anglegeodesic}.)

\begin{figure}
	\centering
	\includegraphics{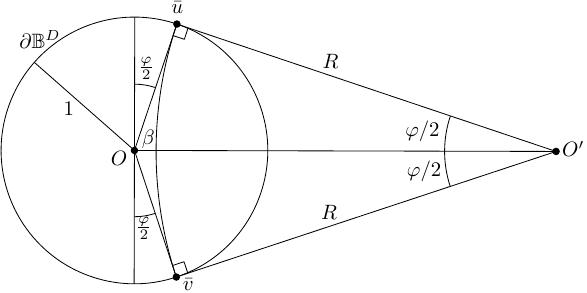}
	\caption{Proof of Lemma~\ref{lem:anglegeodesic}. \label{fig:anglegeodesic}}
\end{figure}
\begin{lemma}\label{lem:anglegeodesic}
	Let $\bar u,\bar v$ be two points on $\partial \poincare^D$ 
	such that $[\bar u,\bar v]$ is at Euclidean distance $\beta$
	from $O$. Let $\varphi=\pi-\angle \bar uO \bar v$. Then we have
		$\pi \beta \leq \varphi \leq 4 \beta.$
\end{lemma}
\begin{proof}
	The geodesic $[\bar u,\bar v]$ is an arc of a circle, let $O'$ be its center and $R$ be its Euclidean radius.
	(See \figurename~\ref{fig:anglegeodesic}.) We apply the law of sines to the triangle
	$OO'\bar v$:
	\[
		\frac{R+\beta}{\sin(\pi/2)}=\frac{1}{\sin(\varphi/2)}=\frac{R}{\cos(\varphi/2)}
	\]
	It follows that $R=1/\tan(\varphi/2)$ and thus, using the change of variable $t=\tan(\varphi/4)$:
	\[
		\beta=\frac{1}{\sin(\varphi/2)}-\frac{1}{\tan(\varphi/2)}=\frac{1-\cos(\varphi/2)}{\sin(\varphi/2)}
		= \frac{2t^2}{2t}=t=\tan(\varphi/4).
	\]
	By concavity of $\tan(\cdot)$, we have $x \leq \tan(x) \leq (4/\pi)x$ for all $0 \leq x \leq \pi/4$. 
	We therefore have $\pi \beta \leq \varphi \leq 4 \beta.$
\end{proof}

\begin{figure}
	\centering
	\includegraphics{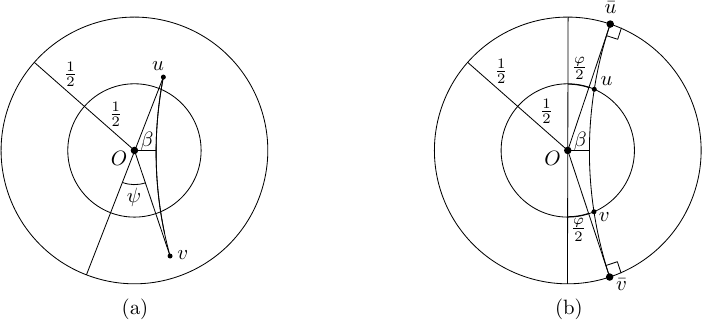}
	\caption{(a) Lemma~\ref{lem:anglegeodesic2} statement. 
		(b) Lemma~\ref{lem:anglegeodesic2} proof.\label{fig:anglegeodesic2}}
\end{figure}

The lemma below is analogous to Lemma~\ref{lem:anglegeodesic}, but
for a geodesic of finite length. (See \figurename~\ref{fig:anglegeodesic2}a.)
\begin{lemma}\label{lem:anglegeodesic2}
	Let $u,v \in \poincare^D$ be two points such that $d_H(O,u) \geq \log 3$ 
	and $d_H(O,v) \geq \log 3$. 
	Let $\psi=\pi-\angle pOq$.
	If $[u,v]$ is at hyperbolic distance $\beta < \log 3$ from $O$, then 
	$\psi < 4 \pi \beta$.
\end{lemma}
\begin{proof}
	The conditions $d_H(O,u) \geq \log 3$ and $d_H(O,v) \geq \log 3$ are
	equivalent to $\|u\| \geq 1/2$ and $\|v\| \geq 1/2$. Similarly, $[u,v]$ is
	at Euclidean distance at most $1/2$ from $O$.

	We extend $[u,v]$ to infinity, such that it is a portion of a geodesic $[\bar u,\bar v]$ with
	$\bar u,\bar v \in \partial \poincare^D$. We change the coordinate system so that $O,\bar u$ and  $\bar v$
	are in the plane $Ox_1x_2$ and $x_1(\bar u)=x_1(\bar v) \geq 0$. 
	When $\bar u$ and $\bar v$ are fixed, the smallest possible angle $\angle uOv$ is 
	achieved when $\|Ou\|=\|Ov\|=1/2$.
	(See \figurename~\ref{fig:anglegeodesic2}b.) So we may assume that  $\|Ou\|=\|Ov\|=1/2$.
	
	We have $x_1(\bar u)=x_1(\bar v)=\sin(\varphi/2)$, using the same notations as in Lemma~\ref{lem:anglegeodesic}.
	It follows that $x_1(u)=x_1(v)<  \sin(\varphi/2)$. We also have $x_1(u)=x_1(v)=(1/2)\sin(\psi/2)$,
	and thus $\sin(\psi/2) < 2 \sin(\varphi/2)$. As $(2/\pi)x \leq \sin x \leq x$ for all $x \in [0,\pi/2]$,
	it implies that $\psi/\pi < \varphi$. The result follows from Lemma~\ref{lem:anglegeodesic}.
\end{proof}

\section{Gromov Hyperbolic Spaces}\label{sec:Gromov}

Let $(X,d)$ be a metric space. The \emph{Gromov product} of two points $u,v \in X$
with respect to $t \in X$ is defined as
\begin{equation*}
(u | v)_{t} = \frac{1}{2} \left( d(u,t) + d(v,t) - d(u,v) \right).
\end{equation*}
This metric space $(X,d)$ is $\delta$-{\it hyperbolic} if for any 4 points $t,u,v,w \in X$, we have
\begin{equation}\label{eq:gromov}
	(u|w)_t \geq \min\{(u|v)_t, (v|w)_t\}- \delta.
\end{equation}

An example of $\delta$-hyperbolic space is $\poincare^D$~\cite[Proposition 4.3]{Coornaert}:

\begin{proposition}\label{prop:HDhyperbolicity}
	The space $(\poincare^D,d_H)$ is $\delta$-hyperbolic, with $\delta=\log 3$.
\end{proposition}

Intuitively, Gromov-hyperbolic spaces behave like tree-metrics. More precisely,
it has been shown that any $n$-points $\delta$-hyperbolic space can be embedded
into a tree metric with $O(\delta \log n)$ additive distortion~\cite{Gromov87}.
Also, a geodesic space is $\delta$-hyperbolic when all its triangles
are thin, as explained in Section~\ref{sec:geogromov}.

The Lemma below was proved by Chepoi et al.~\cite{chepoi2008diameters}. We include a shorter proof.

\begin{lemma}[\cite{chepoi2008diameters}]\label{lem:Chepoi}
	Let $s,t,u,v$ be 4 points in a $\delta$-hyperbolic space $(X,d)$. If \\ $d(s,t) \geq \max\{d(s,u),d(s,v)\}$,
	then $d(u,v) \leq \max\{d(t,u),d(t,v)\}+2\delta$.
\end{lemma}
\begin{proof}
	Without loss of generality, suppose that $(u|t)_s \geq (v|t)_s$. It follows from
	the definition of Gromov hyperbolicity~\eqref{eq:gromov} that
	$
		(u|v)_s 
			 \geq (v|t)_s - \delta.
	$
	It can be rewritten
	$
		d(s,u)+d(s,v)-d(u,v) \geq d(s,v)+d(s,t)-d(t,v) -2 \delta.
	$
	As $d(s,t) \geq d(s,u)$, it implies that $d(u,v) \leq d(t,v)+2\delta$.
\end{proof}

Let $a^*,b^*$ be a diametral pair of $P \subset X$,  $\hat a=f_P(p_1)$ and $\hat b=f_P(\hat a)$.
As was observed by Chepoi et al.~\cite{chepoi2008diameters}, by applying Lemma~\ref{lem:Chepoi} 
with $s=\hat a$, $t=\hat b$, $u=a^*$ and $v=b^*$,
$\hat a$ and  $\hat b$ form an approximate diametral pair in the following sense:
\begin{corollary}\label{cor:Chepoi}
	If $(X,d)$ is $\delta$-hyperbolic, then 
	$d(\hat a,\hat b) \leq \diam(P) \leq d(\hat a,\hat b) + 2\delta$.
\end{corollary}

\subsection{Geodesic $\delta$-Hyperbolic Spaces} \label{sec:geogromov}

\begin{figure}
	\centering
	\includegraphics{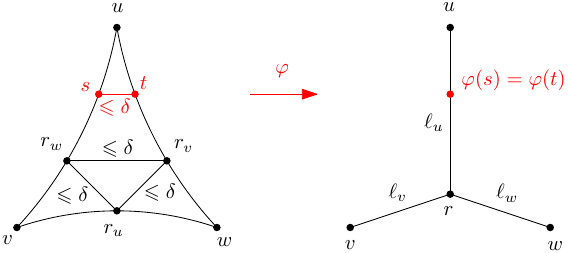}
	\caption{A $\delta$-thin geodesic triangle $[u,v,w]$ (left), 
		and the corresponding tripod (right).\label{fig:slim}}
\end{figure}

Let $u,v,w$ be three points in a geodesic $\delta$-hyperbolic space $(X,d)$.
Let $T$ be a tree, called {\it tripod}, with leafs $u,v,w$, connected to a root node $r$ such that
the length of the three edges are given by (see \figurename~\ref{fig:slim})
\begin{equation*}
d(r,u)=\ell_{u} = (v | w)_{u}, \quad d(r,v)=\ell_{v} = (u | w)_{v}, \text{ and } d(r,w)=\ell_{w} = (u | v)_{w}.
\end{equation*}
It can be easily checked that the distances along this tree coincide with the metric $d$,
for instance $d(u,v)=\ell_u+\ell_v$.

There is a unique map $\varphi:[u,v,w] \to T$ that sends $u,v$ and $w$ to the corresponding leafs
of $T$, and which is an isometry when restricted to each side $[u,v]$, $[v,w]$ and $[u,w]$.
Intuitively, the geodesic triangle is folded onto the tripod, so that each point in the interior of
an edge of the tripod corresponds to two points of the triangle, taken from the two adjacent sides.

The triangle $[u,v,w]$ is $\delta'$-\emph{thin} if for any two points $s,t \in [u,v,w]$, 
$\varphi(s) = \varphi(t)$ implies $d(s,t) \leq \delta'$. It has been shown that in a Gromov-hyperbolic
geodesic space, all triangles are thin~\cite[Proposition 6.3C]{Gromov87} :

\begin{proposition}\label{prop:thintriangles}
	If $(X,d)$ is a geodesic $\delta$-hyperbolic metric space, then all the geodesic
	triangles in $(X,d)$ are $2\delta$-thin.
\end{proposition}

In the case of $\poincare^D$, the following is known~\cite[Corollary 4.2]{Coornaert}.
\begin{proposition}\label{prop:HDthin}
	Every triangle in $(\poincare^D,d_H)$ is $(\log 3)$-thin.
\end{proposition}

\subsection{A Property of Approximate Diametral Pairs}

We use the same notations as above, so 
$X$ is a  geodesic $\delta$-hyperbolic space, $P \subset X$ and $p_1 \in P$. Corollary~\ref{cor:Chepoi} shows that a pair $\hat a=f_P(p_1)$, $b=f_P(\hat a)$ is an approximate diametral pair. We will prove (Lemma~\ref{lem:farthest2}) that the midpoint $\hat m$
of $[\hat a,\hat b]$ is close to any geodesic $[q,f_P(q)]$. We first consider the midpoint $m^*$ of an exact diametral pair $a^*,b^*$.

\begin{lemma}\label{lem:farthest:main}
Let $(X,d)$ be a geodesic $\delta$-hyperbolic space such that every geodesic
triangle is $\delta'$-thin.
Let $q \in X$, $P \subset X$ and $q'=f_P(q)$.
Then the point $m$ along the geodesic segment $[q,q']$ 
that is at distance $\diam(P)/2-\delta$ from $q'$ satisfies $d(m,m^*) \leq \delta+2\delta'$.
\end{lemma}
\begin{proof}
	Without loss of generality, suppose that $d(a^*,q') \leq d(b^*,q')$. 	Then by Lemma~\ref{lem:Chepoi},
	\begin{equation}\label{eq:th:farthest1}
		d(b^*,q') \leq \diam(P) \leq d(b^*,q') +2\delta.
	\end{equation}
	Let $m$ and $m_1$ be the points at distance $\diam(P)/2-\delta$ from $q'$ 
	along a geodesic $[q',q]$ and a geodesic $[b^*,q']$, respectively. 
	As $d(q,q') \geq d(b^*,q)$, it follows that
	\begin{align*}
		(b^*|q)_{q'}&=\left(d(b^*,q')+d(q,q')-d(b^*,q) \right)/2\\ 
			& \geq  \diam(P)/2-\delta.
	\end{align*}
	As the triangle $[b^*,q,q']$ is $\delta'$-thin, it follows that 
	\begin{equation}\label{eq:mm1}
		d(m,m_1) \leq \delta'.
	\end{equation}

	We now consider the geodesic triangle $[a^*,b^*,q']$. 
	\begin{align*}
		(a^*|q')_{b^*} &= \left(d(a^*,b^*)+d(b^*,q')-d(a^*,q') \right)/2\\
			& =\left(\diam(P)+d(b^*,q')-d(a^*,q') \right) /2 \\
			& \geq \diam(P) /2 
	\end{align*}
	Let $m_2$ be the point along $[b^*,q']$ at distance $\diam(P)/2$ from $b^*$.
	As the triangle $[a^*,b^*,q']$ is $\delta'$-thin, we have 
	\begin{equation}\label{eq:m2m*}
		d(m_2,m^*) \leq \delta'.
	\end{equation}
	
	The points $m_1$ and $m_2$ are both along $[b^*,q']$, with $m_1$ being at distance
	$\diam(P)/2-\delta$ from $q'$ and $m_2$ being at distance $\diam(P)/2$ from $b^*$.
	So we have
	\[
		d(m_1,m_2) = |\diam(P)/2-\delta-d(b^*,q')+\diam(P)/2|=|\diam(P)-d(b^*,q')-\delta|.
	\]
	By inequality~\eqref{eq:th:farthest1}, it follows that $d(m_1,m_2) \leq \delta$.
	Then $d(m,m^*) \leq \delta+2\delta'$ follows from inequalities~\eqref{eq:mm1} and~\eqref{eq:m2m*}.
\end{proof}

We will not be able to use directly the lemma above in our construction, because we cannot compute the midpoint $m^*$ of an exact diametral pair in linear time. The lemma below allows us
to use the approximate midpoint $\hat m$, which can be computed is linear time.

\begin{lemma}\label{lem:farthest2}
	Let $(X,d)$ be a geodesic $\delta$-hyperbolic space 
	such that every geodesic triangle is $\delta'$-thin.
	Let $q \in X$, $P \subset X$ and
	$q'=f_P(q)$. Then the point $m$ along the geodesic segment $[q,q']$ 
	that is at distance $\diam(P)/2-\delta$ from $q'$ satisfies $d(m,\hat m) \leq 3\delta+4\delta'$.
	In particular, if $(X,d)=(\poincare^D,d_H)$, we have 
	$d(m,\hat m) \leq 7 \log 3 < 8$.
\end{lemma}
\begin{proof}
	Let $m_3$ be the point of $[\hat a,\hat b]$ that is at distance $\diam(P)/2-\delta$
	from $\hat b$. By Lemma~\ref{lem:farthest:main}, we have $d(m_3,m^*) \leq \delta+2\delta'$ 
	and thus $d(m,m_3) \leq 2\delta+4\delta'$. 
	
	The points $m_3$ and $\hat m$ are along $[\hat a,\hat b]$, with $m_3$ at distance
	$\diam(P)/2-\delta$ and $d(\hat a,\hat b)/2$ from $\hat b$, respectively. So we have
	\[
		d(m_{3}, \hat m)=|d(\hat a,\hat b)+2\delta-\diam(P)|/2 
	\]
	which is at most $\delta$ by Corollary~\ref{cor:Chepoi}. It follows that $d(m,\hat m) \leq 3\delta+4\delta'$.
\end{proof}

\section{Coreset for Large Diameter Input}\label{sec:far}

We now consider farthest-neighbor searching when the input point set has diameter
larger than 5. So $P=\{p_1,\dots,p_n\}$ is a subset of $\poincare^D$, and 
we assume that $\Delta_1 \geq 5$, where $\Delta_1=d_{H} (p_1,f_P(p_1))$, 
which implies that $\diam(P) \geq 5$.
We describe below the construction of our coreset $P_\eps$ for farthest-point queries. 

We first compute an approximate diametral pair $\hat a, \hat b$, where $\hat a=f_P(p_1)$
and $\hat b=f_P(\hat a)$. Let  $\hat m$ be the midpoint of $[\hat a,\hat b]$. Without loss of generality, we assume that $\hat m=O$, as otherwise, we can apply the isometry $\mu_{\hat{m}}$ to $P$.

In the same way as we did for the small diameter case (Proposition~\ref{prop:grid}),
we construct a set $G$ of $O(1/\eps^D)$ grid points in the hyperbolic ball
$B_0=B_H(O,8 \log 3)$ such that for any point $p \in B_0$, there is a point $g \in G$
that satisfies $d_H(p,g) \leq \eps/2$. For each point in $g \in G$, we insert
$f_P(g)$ into our coreset $P_\eps$. There are $O(1/\eps^D)$ such points.

Let $S_0$ be the sphere centered at $O=\hat m$ with hyperbolic radius $7 \log 3$.
We construct a set $H \subset S_0$, called a $\sqrt{\eps}/(8\pi)$-net, 
of size $O(1/\eps^{(D-1)/2})$. It has the property that for any point $p \in S_0$, there 
is a point $h \in H$ such that $d_H(p,h) \leq \sqrt\eps/(8\pi)$. Such a set can be easily constructed, for instance by constructing a regular grid of mesh $\Theta(\sqrt \eps)$ on the box circumscribed to $S_0$, and projecting its vertices onto $S_0$. 

For each point $h \in H$, we apply $\mu_{h}$ to all the points in $P$, which takes $O(n)$ time. After this transformation, we have $h=O$. Let $\mathcal C$ be a partition of the space around $O$ into $O(1/\eps^{(D-1)/2})$ simplicial cones of angular diameter $\sqrt{\eps}/2$. This partition can be easily constructed using a grid circumscribed to $S_0$. (See a more detailed description in the book by Narasinham and Smid~\cite{narasimhan_smid_2007}.)

For each cone $C \in \mathcal C$ that contains 
at least one point of $\mu_h(P)$, we
insert into our coreset $P_\eps$ a point $p_\eps \in P$ such that $\mu_h(p_\eps)$ is a point
in $C \cap \mu_h(P)$ that is farthest from $O$. For each $h \in H$, 
there are $O(1/\eps^{(D-1)/2})$ such points, so the size of $P_\eps$ remains $O(1/\eps^D)$.

It remains to prove that this construction is correct. So we need to argue that
for any query point $q \in \poincare^D$, there is a point $p_\eps \in P_\eps$ such
that $d_H(q,p_\eps) \geq d_H(q,f_P(q))-\eps$. 

There are two cases. First, suppose that $q \in B_0$. Then there is a point $g \in G$
such that $d_H(q,g) \leq \eps/2$, and the point $g'=f_P(g)$ is in $P_\eps$.
It follows that
\begin{align*}	
	d_H(q,g') & \geq d_H(g,g')-d_H(q,g) \\
		& \geq d_H(g,f_P(q)) - \eps/2 \\
		& \geq d_H(q,f_P(q)) - d_H(q,g) -\eps/2 \\
		& \geq d_H(q,f_P(q)) - \eps.
\end{align*}
So we can take $p_\eps=g'$.

Now suppose that $q \notin B_0$. Let $q'=f_P(q)$, and let $m$ be the point along
$[q,q']$ that is at distance $\diam(P)/2-\log 3$ from $q'$.
By Lemma~\ref{lem:farthest2}, this point $m$ is inside $B_H(O,7 \log 3)$, and thus
the geodesic $[q,q']$ crosses $S_0$ at a point $s$ between $q$ and $m$.
Then we must have $d_H(q',s) \geq \diam(P)/2-\log 3 > 5/2-\log 3$.

There is a point $h \in H$ such that $d_H(s,h) \leq \sqrt{\eps}/4$. As $q \notin B_0$
and $h \in H$, 
we must have $d_H(q,h) \geq \log 3$. As $d_H(q',s) > 5/2-\log 3$ and $\eps<1$,
we also have $d_H(q',h) \geq \log 3$. 

Let $q_h=\mu_h(q)$ and $q'_h=\mu_h(q')$. As $q$ and $q'$ are at hyperbolic distance at least 
$\log 3$ from $h$, we have $d_H(O,q_h) \geq \log 3$ and $d_H(O,q'_h) \geq \log 3$.
In addition, the geodesic $[q_h,q'_h]$ goes through $\mu_h(s)$, which is at distance at most
$\sqrt{\eps}/(8\pi)$ from $O$. So by Lemma~\ref{lem:anglegeodesic2}, the 
angle $\psi=\pi-\angle q_h O q'_h$ satisfies $\psi \leq \sqrt{\eps}/2$.

Let $C$ be the cone in $\mathcal C$ that contains $q'_h$. Let $c$ denote the corresponding
point that we inserted into $P_\eps$, so $c$ is a point in $C \cap \mu_h(P)$ that is furthest
from $O$. In particular, we have $d_H(O,c) \geq d_H(O,q'_h)$, 
and there is a point $p_\eps \in P_\eps$
such that $c=\mu_h(p_\eps)$. As $c$ and $q'_h$ are in the same cone, the angle 
$\psi'=\pi-\angle q_hOc$ satisfies $\psi' \leq \psi+\sqrt{\eps}/2 \leq \sqrt{\eps}$.
Then we have
\begin{align*}
	d_H(q,p_\eps) &= d_H(q_h,c) \\
		& \geq d_H(O,q_h)+d_H(O,c) - \eps & \text{by lemma~\ref{lem:theta}} \\
		& \geq d_H(O,q_h)+d_H(O,q'_h) - \eps \\
		& \geq d_H(q_h,q'_h)-\eps \\
		& = d_H(q,q') - \eps \\
		& = d_H(q,f_P(q)) - \eps.
\end{align*}

So we just proved the following:
\begin{lemma}\label{lem:far}
	Let $P=\{p_1,\dots,p_n\}$ be a set of $n$ points in $\poincare^D$ 
	such that $\Delta_1=d_H(p_1,f_P(p_1))$ satisfies $\Delta_1 \geq 5$.
	For any $\eps$ such that $0 < \eps < 1$,
	we can construct in $O\left(n/\eps^D\right)$ time a 
	coreset $P_\eps \subset P$ for farthest-point queries of 
	size $|P_\eps|=O(1/\eps^D)$. In particular, for any query point $q \in \poincare^D$,
	there is a point $p_{\eps} \in P_{\eps}$ such that 
	$
		d_H(q,p_{\eps}) \geq d_H(q,f_P(q))-\eps.
	$
\end{lemma}

We can now combine the result above for large diameter input with our result for small diameter.
When $\Delta_1 \leq 5$, we construct a coreset $P_\eps'$ with
$\eps'=\eps/10$ according to Lemma~\ref{lem:small}.
In this case, $\diam(P) \leq 5$, so it gives a relative error $\eps$.
When $\Delta_1 \geq 5$, we apply Lemma~\ref{lem:far}. As $\diam(P) \geq 5$, it gives
an additive error at most $\eps$. So we obtain the following:

\begin{theorem}\label{th:faralgo}
	Let $D$ be a fixed integer and $0 < \eps < 1$.
	Let $P$ be a set of $n$ points in $\poincare^D$.
	We can construct in $O\left(n/\eps^D\right)$ time a 
	coreset $P_\eps \subset P$ for farthest-point queries of 
	size $|P_\eps|=O(1/\eps^D)$. In particular, for any query point $q \in \poincare^D$,
	there is a point $p_{\eps} \in P_{\eps}$ such that 
	$
		d_H(q,p_{\eps}) \geq d_H(q,f_P(q))-\eps
	$
	and $d_H(q,p_{\eps}) \geq (1-\eps)d_H(q,f_P(q))$.
\end{theorem}

%
\bibliographystyle{plain}
\bibliography{hdiameter}

\begin{thebibliography}{10}

\bibitem{agarwal1992farthest}
Pankaj~K Agarwal, Ji{\v{r}}{\'\i} Matou{\v{s}}ek, and Subhash Suri.
\newblock Farthest neighbors, maximum spanning trees and related problems in
  higher dimensions.
\newblock {\em Computational Geometry}, 1(4):189--201, 1992.

\bibitem{Arya98}
Sunil Arya, David~M. Mount, Nathan~S. Netanyahu, Ruth Silverman, and Angela~Y.
  Wu.
\newblock An optimal algorithm for approximate nearest neighbor searching fixed
  dimensions.
\newblock {\em J. ACM}, 45(6):891--–923, 1998.

\bibitem{berg2008computational}
Mark~de Berg, Otfried Cheong, Marc~van Kreveld, and Mark Overmars.
\newblock Computational geometry: Algorithms and applications, 2008.

\bibitem{chepoi2008diameters}
Victor Chepoi, Feodor Dragan, Bertrand Estellon, Michel Habib, and Yann
  Vax{\`e}s.
\newblock Diameters, centers, and approximating trees of $\delta$-hyperbolic
  geodesic spaces and graphs.
\newblock In {\em Proc. 24th annual symposium on Computational geometry}, pages
  59--68, 2008.

\bibitem{Clarkson88}
Kenneth~L. Clarkson.
\newblock A randomized algorithm for closest-point queries.
\newblock {\em {SIAM} J. Comput.}, 17(4):830--847, 1988.

\bibitem{Coornaert}
Michel Coornaert, Thomas Delzant, and Athanase Papadopoulos.
\newblock Les groupes hyperboliques de gromov.
\newblock In {\em Geometrie et theorie des groupes}, volume 1441 of {\em
  Lecture Notes in Mathematics}. Springer, 1990.

\bibitem{BergT24}
Mark de~Berg and Leonidas Theocharous.
\newblock A coreset for approximate furthest-neighbor queries in a simple
  polygon.
\newblock In {\em Proc. 40th International Symposium on Computational
  Geometry}, pages 16:1--16:16, 2024.

\bibitem{NEURIPS2018}
Octavian Ganea, Gary Becigneul, and Thomas Hofmann.
\newblock Hyperbolic neural networks.
\newblock In {\em Advances in Neural Information Processing Systems}, 2018.

\bibitem{gezalyan}
Auguste~H. Gezalyan and David~M. Mount.
\newblock {Voronoi Diagrams in the Hilbert Metric}.
\newblock In {\em Proc. 39th International Symposium on Computational
  Geometry}, pages 35:1--35:16, 2023.

\bibitem{Gromov87}
M.~Gromov.
\newblock Hyperbolic groups.
\newblock In {\em Essays in group theory}, volume~8 of {\em math. sci. res.
  ins. publ.}, pages 75--263. Springer, 1987.

\bibitem{HPbook}
Sariel Har-peled.
\newblock {\em Geometric Approximation Algorithms}.
\newblock American Mathematical Society, 2011.

\bibitem{kisfaludi2024quadtree}
S{\'a}ndor Kisfaludi-Bak and Geert van Wordragen.
\newblock A quadtree, a steiner spanner, and approximate nearest neighbours in
  hyperbolic space.
\newblock In {\em Proc. 40th International Symposium on Computational
  Geometry}, pages 68:1--68:15, 2024.

\bibitem{krauthgamer2006algorithms}
Robert Krauthgamer and James~R Lee.
\newblock Algorithms on negatively curved spaces.
\newblock In {\em 47th Annual IEEE Symposium on Foundations of Computer
  Science}, pages 119--132, 2006.

\bibitem{Nature}
Alessandro Muscoloni, Josephine~Maria Thomas, Sara Ciucci, Ginestra Bianconi,
  and Carlo~Vittorio Cannistraci.
\newblock Machine learning meets complex networks via coalescent embedding in
  the hyperbolic space.
\newblock {\em Nature Communications}, 8:1615, 2017.

\bibitem{narasimhan_smid_2007}
Giri Narasimhan and Michiel Smid.
\newblock {\em Geometric Spanner Networks}.
\newblock Cambridge University Press, 2007.

\bibitem{pagh2015approximate}
Rasmus Pagh, Francesco Silvestri, Johan Sivertsen, and Matthew Skala.
\newblock Approximate furthest neighbor in high dimensions.
\newblock In {\em 8th International Conference on Similarity Search and
  Applications}, pages 3--14, 2015.

\bibitem{park2025embeddings}
Eunku Park and Antoine Vigneron.
\newblock Embeddings and near-neighbor searching with constant additive error
  for hyperbolic spaces.
\newblock {\em Computational Geometry}, 126:102150, 2025.

\bibitem{ratcliffe1994foundations}
John~G Ratcliffe, Sheldon Axler, and Kenneth~A Ribet.
\newblock {\em Foundations of hyperbolic manifolds}, volume 149.
\newblock Springer, 1994.

\bibitem{1354510}
Y.~Shavitt and T.~Tankel.
\newblock On the curvature of the internet and its usage for overlay
  construction and distance estimation.
\newblock In {\em Proc. IEEE INFOCOM}, page 384, 2004.

\bibitem{SongJA25}
Minju Song, Mook~Kwon Jung, and Hee{-}Kap Ahn.
\newblock Farthest-point voronoi diagrams in the hilbert metric.
\newblock In {\em Proc. 19th International Symposium on Algorithms and Data
  Structures}, pages 48:1--48:15, 2025.

\end{thebibliography}

\end{document}